\documentclass{article}

\usepackage{arxiv}

\usepackage[utf8]{inputenc} % allow utf-8 input
\usepackage[T1]{fontenc}    % use 8-bit T1 fonts
\usepackage{hyperref}       % hyperlinks
\usepackage{url}            % simple URL typesetting
\usepackage{booktabs}       % professional-quality tables
\usepackage{amsfonts}       % blackboard math symbols
\usepackage{nicefrac}       % compact symbols for 1/2, etc.
\usepackage{microtype}      % microtypography
\usepackage{lipsum}
\usepackage{graphicx}
\usepackage{booktabs}
\usepackage{listings, array, blkarray}
\usepackage{bbm, dsfont}
\usepackage{multicol}
\usepackage{wrapfig,multirow}
\usepackage{amsmath,amsthm}
\usepackage{algorithm}
\usepackage{algorithmic}

\usepackage[caption=false]{subfig}
\newcommand{\denselist}{\itemsep 0pt\partopsep 0pt}

\newcommand{\matindex}[1]{\mbox{\scriptsize#1}}
\newtheorem{theorem}{Theorem}

\newtheorem{lemma}{Lemma}

\newtheorem{innercustomprob}{Problem}

\makeatletter
\newcommand\footnoteref[1]{\protected@xdef\@thefnmark{\ref{#1}}\@footnotemark}
\makeatother

\newenvironment{customprob}[1]
  {\renewcommand\theinnercustomprob{#1}\innercustomprob}
  {\endinnercustomprob}

\title{Who and When to Screen: Multi-Round Active Screening for Recurrent Infectious Diseases Under Uncertainty}

\author{
  Han-Ching Ou \\
  University of Southern California\\
  Los Angeles, California, USA\\
  \texttt{hanchino@usc.edu} \\
   \And
  Arunesh Sinha \\
  University of Michigan\\
  Ann-Arbor, Michigan, USA\\
   \And
  Sze-Chuan Suen, \\ 
  \textbf{Andrew Perrault}, \\
  \textbf{Milind Tambe} \\
  University of Southern California\\
  Los Angeles, California, USA\\
}

\begin{document}
\maketitle
\begin{abstract}

Controlling recurrent infectious diseases is a vital yet complicated problem.
In this paper, we propose a novel active screening model (ACTS) and algorithms to facilitate active screening for recurrent diseases (no permanent immunity) under infection uncertainty. Our contributions are: (1) A new approach to modeling multi-round network-based screening/contact tracing under uncertainty, which is a common real-life practice in a variety of diseases~\cite{eames2003contact,world2013systematic}; (2) Two novel algorithms, \textsc{Full-} and \textsc{Fast-REMEDY}. \textsc{Full-REMEDY} considers the effect of future actions and finds a policy that provides high solution quality, where \textsc{Fast-REMEDY} scales linearly in the size of the network; (3) We evaluate \textsc{Full-} and \textsc{Fast-REMEDY} on several real-world datasets which emulate human contact and find that they control diseases better than the baselines. To the best of our knowledge, this is the first work on multi-round active screening with uncertainty for diseases with no permanent immunity.

\end{abstract}

% AAMAS: the ACM CCS are not needed within AAMAS papers
%%
%% The code below should be generated by the tool at
%% http://dl.acm.org/ccs.cfm
%\begin{CCSXML}
%<ccs2012>
%  <concept_id>10010520.10010553.10010562</concept_id>
%  <concept_desc>Computer systems organization~Embedded systems</concept_desc>
%  <concept_significance>500</concept_significance>
%  <concept_id>10010520.10010575.10010755</concept_id>
%  <concept_desc>Computer systems organization~Redundancy</concept_desc>
%  <concept_significance>300</concept_significance>
%  <concept_id>10010520.10010553.10010554</concept_id>
%  <concept_desc>Computer systems organization~Robotics</concept_desc>
%  <concept_significance>100</concept_significance>
%  <concept_id>10003033.10003083.10003095</concept_id>
%  <concept_desc>Networks~Network reliability</concept_desc>
%  <concept_significance>100</concept_significance>
%</ccs2012>  
%\end{CCSXML}
%
%\ccsdesc[500]{Computer systems organization~Embedded systems}

  % put your semicolon-separated keywords here!

\maketitle

%%%%%%%%%%%%%%%%%%%%%%%%%%%%%%%%%%%%%%%%%%%%%%%%%%%%%%%%%%%%%%%%%%%%%%%%%%%%%%%%%%%%%%%%%%%%%%%%%%%%%%%%%
%% start of main body of paper

\section{Introduction}
Contagious diseases, such as influenza, gonorrhea, and chlamydia are critical public-health challenges that continue to threaten lives and economic productivity. While low-cost treatment programs are available, individuals may ignore symptoms and delay care, increasing transmission risk. Public health agencies may therefore engage in active screening or contact tracing efforts, where individuals in the community are asked to undergo diagnostic tests and are offered treatment if tests return positive results~\cite{eames2003contact,jama1984}. 

However, in many settings, active screening/contact tracing is expensive and time consuming. 
%In India, an estimated 1 million missing TB cases require an efficient method of active screening, particularly given limited health budgets \cite{highRisk}. 
Even in the United States, budget cuts in 52\% of states and STD programs in 2012 has impacted the quality and quantity of the contact tracing~\cite{braxton2017sexually}. Efficiently identifying infectious cases is therefore of vital importance.

Our \emph{first contribution} is a model of the active screening problem (ACTS). We focus on recurrent infectious diseases that assumes contact to be in structured networks with the SIS model of transmission~\cite{wang2003epidemic}, which is applicable for a wide range of diseases such as pertussis, syphilis and typhoid. The SIS model is the foundation of more complex models that capture more disease dynamics (such as latent states, variation in birth/death rates, or multiple treatment states). In the SIS model, individuals can be either susceptible (S) (currently healthy, but may become exposed) or infected (I). We consider diseases for which there is no means to achieve permanent immunity and prove the ACTS problem to be NP-hard. We assume that health workers are uncertain about the health state of individuals, have a small budget relative to population size for active screening and must engage in active screening over multiple rounds (time periods). To the best of our knowledge, no other models consider multi-round active screening for network SIS diseases in the AI literature.

% Change to bullet points
% Contributions: 1. New intervene problem. 2. Two algorithm, one focus on solution quality second is fast without sacrificing a lot. Potential 3.: belief update isn't new, only the conditional probability (acquire new information for each cure individual) is novel, don't know if we should make a point for that.
% Add description of Full vs Fast, trade off looks good when using true state, but effect is limit when using belief.

Our \emph{second contribution} is a novel algorithm---\textbf{RE}current screening \textbf{M}ulti-round \textbf{E}fficient \textbf{DY}namic algorithm (\textbf{REMEDY})---to guide scalable active screening. We develop two versions of the algorithm, \textsc{Full-} and \textsc{Fast-REMEDY}. In \textsc{Full-REMEDY}, we consider both current and future actions simultaneously to understand the underlying disease dynamics and uncertainty of individuals' health states. \textsc{Fast-REMEDY} reduces the time complexity to scale to very large networks by exploiting eigendecomposition techniques. We illustrate the benefits of \textsc{Full-} and \textsc{Fast-REMEDY} via extensive testing on a set of real-world human contact networks against various baselines across a range of realistic disease parameters.

The paper is structured as follows. In Section~\ref{sec:diseasemodel}, we provide background on disease modeling and active screening. In Section~\ref{sec:acts}, we formalize the active screening problem (ACTS) and prove that it is NP-hard. In Section~\ref{sec:alg}, we present REMEDY, our novel algorithm for ACTS. In Section~\ref{sec:experiments}, we empirically analyze REMEDY and compare its performance to relevant baselines.

\section{Disease Model and Background}
\label{sec:diseasemodel}

We first introduce the disease model notation for our problem. In an SIS model~\cite{anderson1992infectious,bailey1975mathematical}, an individual can either be in state $S$ (a healthy individual \textit{susceptible} to disease) or $I$ (the individual is \textit{infected}). 
%We do not consider an explicit recovered or permanent immunity state ($R$) in our model, as this has been the focus of prior studies~\cite{contact1,contact2}.
SIS models recurrent diseases, where permanent immunity is not possible (e.g., TB, typhoid) and not diseases such as Hepatitis A and measles, which follow a SIR or SEIR pattern where treated individuals may achieve permanent immunity by entering $R$ state.
%or SEIS~\cite{seis1} disease dynamics. 

\subsection{Disease Model}
We adopt a discrete time SIS model for modeling the disease dynamics, which was earlier considered by Wang et al.~\cite{wang2003epidemic}. Given a contact network $G(V,E)$, infection spreads via the edges in the network. There are $|V|$ individuals, and let $\delta(v)$ denote neighbors of node $v$ in the network. Each individual (node) $v$ in the network (at time t) is in state ${\bf s}_v(t)\in\{S, I\}$. Let ${\bf t}_v(t)$ denote the state vector that represents the true state of node $v$ at time $t$ where $S$ is represented as $[1,0]^{\top}$ and $I$ as $[0,1]^{\top}$. Given the initial state, a infected node infects its healthy neighbors with rate $\alpha$ independently and recovers with probability $c$. The health state transition probabilities of a node is given by $P  \left[s_v(t+1)=\{S, I\} \right]=\mathbf{T}^N_v(t) {\bf t_v}(t)$ where
\begin{gather}
\mathbf{T}^N_v(t) = \begin{blockarray}{ccc}
  & \matindex{$S$} & \matindex{$I$}\\
    \begin{block}{c[cc]}
      \matindex{$S$} & 1-q_v & c \\
      \matindex{$I$} & q_v & 1-c \\
    \end{block}
  \end{blockarray} \; ,
\end{gather}
where $q_v = 1-(1-\alpha)^{\lvert \{u \in \delta(v) ~|~ {\bf s_u}(t) = I\}\rvert}$. Note that the columns denote the state of $v$ at time $t$ and the rows denote the state at $t+1$. The transition probabilities follow the disease dynamics described earlier. In particular, $q_v$ captures the exact probability that node $v$ becomes infected from its infected neighbors $\{u \in \delta(v) ~|~ {\bf s_u}(t) = I\}$ and $c$ captures the probability that $I$ individuals recover without active screening.
%seek treatment voluntarily. 

Given such transition probabilities and an initial state, if no intervention happens, the network state evolves by flipping biased coins for each node to determine their next true state in each time step. The process is repeated until the terminal step $T$ is reached.

\subsection{Prior Approaches to Active Screening}

Most previous work on active screening deals primarily with SIR or SEIR type diseases, often referred to as the \textit{Vaccination Problem}~\cite{contact1,contact2,wangthesis,dava,ganesh2005effect}, where permanent immunization (entry into $R$ state) can be viewed as removing nodes from the graph~\cite{eigen1,eigen3}.
Exploiting this idea, Saha et al.~\cite{eigen3} and Tong et al.~\cite{eigen1} focus on immunization ahead of an epidemic and suggest a heuristic method of removing a set of $k$ nodes based on the eigenvalues of the adjacency matrix. Zhang and Prakash~\cite{dava} consider the problem of selecting the best $k$ nodes to immunize in a network after the disease has started to spread. These methods assume that the exact status of each node is known and deal with a single round of screening that offers permanent immunity. 

However, for diseases in which there is no permanent immunity, one-time screening (cure) is not enough and, further, it may not be reasonable to quarantine patients until the disease has died out. We focus on a multi-round screening of SIS diseases that cannot be permanently cured. To the best of our knowledge, this complex setting has not been studied previously. Generally, the problem of minimizing disease spread is different from the well-studied problem of influence maximization~\cite{kempe2000,chen2009}, where one optimizes the selection of seeds or starting nodes for maximizing influence spread, as opposed to optimizing the selection of nodes on which to intervene in order to minimize disease spread.

\section{The Active Screening (ACTS) Problem}\label{sec:acts}
Motivated by active screening/contact tracing campaigns that has been practiced since the 1980s~\cite{jama1984} and applied in various forms/diseases~\cite{ebola,braxton2017sexually}, we propose the Active Screening (ACTS) Problem. Given the SIS model in previous section, an active screening agent seeks to determine the best node sets $C_a(t) \subset V$ to actively screen and cure with a limited budget of $|C_a(t)|\leq k$ at each time step $t$. The agent does not know the the ground truth health state of all individuals. The agent knows the ground truth of the network structure $G(V,E)$, the infection probability $\alpha$ and recovery probability $c$. In addition, the agent observes the \emph{naturally cured} node set $C_n(t)$ at time $t$---because this set of patients come to the clinic voluntarily. The active screening happens after the agent acquires information about $C_n(t)$. Let $C_a(t)$ be the set of nodes that are actively screened at time $t$. A node $v \in C_n(t)\cup C_a(t)$ becomes cured at time $t+1$.
Thus, the transition matrix for a node $v\in C_n(t)\cup C_a(t)$ is $P \left[ {\bf t_v}(t+1)=\{S, I\} \right]=\mathbf{T}^A_v(t) {\bf s_v}(t)$, where
\begin{gather}
\mathbf{T}^A_v(t) = \begin{blockarray}{ccc}
  & \matindex{$S$} & \matindex{$I$}\\
    \begin{block}{c[cc]}
      \matindex{$S$} & 1 & 1 \\
      \matindex{$I$} & 0 & 0 \\
    \end{block}
  \end{blockarray} \; .
\end{gather}
It is worth noting that the action the agent takes at time $t$ does not affect the transition matrix $\mathbf{T}^N_v(t)$ of the nodes not involved in active screening.

Unlike most of the previous work which focused on minimizing the budget spent until the long-term disease eradication is achieved, we focus on maximizing the health quality of each individual at each time step. We treat each individual and time step equally, although it is easy to modify the costs and values to be weighted. The objective of the ACTS problem is:
\begin{align}\label{eq:obj}
\min_{C_a(0),\ldots,C_a(T)} \mathbb{E}\left[\sum\nolimits_{t=0}^{T}\sum\nolimits_{v\in V} | {\bf s_v}(t)=I| \right]. 
\end{align}

\begin{customprob}{Statement}
%(ACTS Problem) Given a contact network $G(V,E)$, disease transition dynamics $\mathbf{T}^N, \mathbf{T}^A$, time horizon $T$, and budget $k$ for each time step, find an active screening policy such that the expectation of $\sum_{t=0}^{T}\sum_{v\in V} |{\bf s_v}(t)=I|$ is minimized.
(ACTS Problem) Given a contact network $G(V,E)$, the disease and active screening model, find an active screening policy such that the expectation of $\sum_{t=0}^{T}\sum_{v\in V} |{\bf s_v}(t)=I|$ is minimized.
\end{customprob}
Even assuming we know the ground truth infected state for each node, ACTS is NP-hard.
\begin{theorem} \label{NPhard}
The ACTS Problem is NP-hard.
\end{theorem}
\begin{proof}
We reduce the \textsc{VertexCover} to the decision problem ``Does there exist a curing strategy of objective function equals $5|V|-2k$ with budget of k each round of the constructed ACTS problem?''

Given a \textsc{VertexCover} decision problem with graph $G=(V,E)$ and budget $k$, we construct a new graph $G^*=(V^*_0\cup V^*_1 \cup V^*_2,E^*)$ as follows: First, for each node $v \in V$, create three nodes $v_0$, $v_1$ and $v_2$ in $G^*$. Second, for each node $v \in V$, create an edge $(v_0, v_1)$ in $G^*$. Finally, for each edge $(u,v) \in E$ create two edges, $(u_1,v_2)$ and $(u_2, v_1)$ in $G^*$. We set the parameters of the ACTS problem to be $(\alpha, c)=(1,0)$ and $T=2$ with budget of $k$ in each round. The initial state of the graphs are $s_v(0)=I$ $\forall$ $v \in V^*_0$ and $s_v(0)=S$ $\forall$ $v \in V^*_1 \cup V^*_2$. Fig.~\ref{NPhardG} shows a simple example.
\begin{figure}[ht]
\includegraphics[width=0.30\textwidth,keepaspectratio]{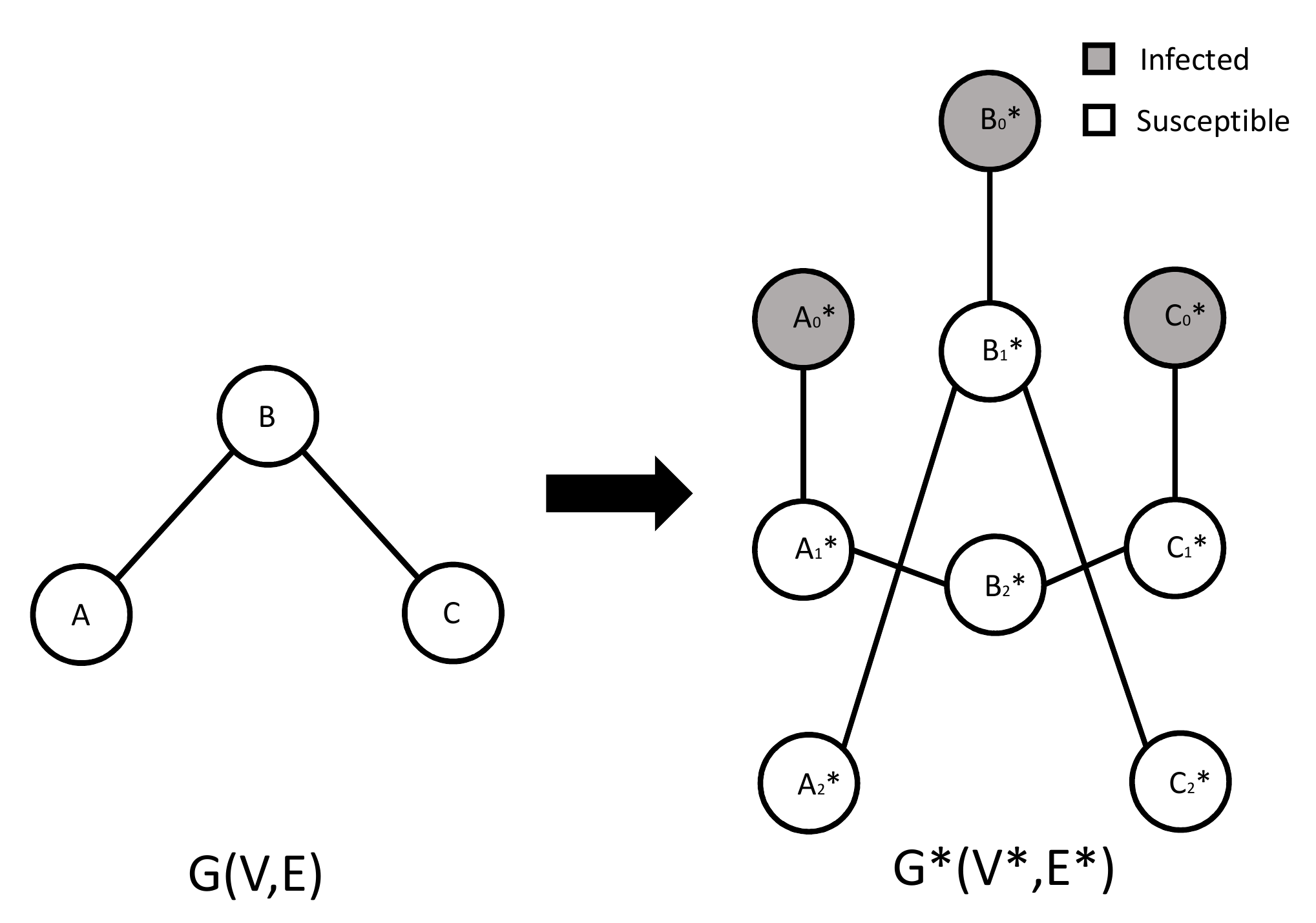}%
\centering
\caption{A simple example of graph transformation for problem deduction.}\label{NPhardG}
\end{figure}

We now argue that $G$ has a vertex cover of size $k$ if and only if the ACTS problem of the above setting has the objective function of $5|V|-2k$. In the above setting, we get to act twice. Acting at $t=0$ allows us to force $k$ nodes into $S$ state at $t=1$. Denote the objective function we get at time $t$ as $Score(t)$, no matter what nodes we chose at $t=0$, our sum of score in the first two rounds is always going to be $Score(0)=|V|$, $Score(1)=2|V|-k$ and for the action we take at $t=1$ will only reduce $Score(2)$ by amount of $k$, as long as we pick nodes in $I$ state since it has no chance to propagate. Thus the only action matters is the action on $t=0$ toward $Score(2)$. Picking the copy of vertex cover set of $G$ in $V^*_1$ results in $|V|+(|V|-k)+k$ of $I$ nodes in $t=2$, which are all the nodes in $V^*_0$, all the nodes in $V^*_1$ except vertex cover copy and the vertex cover in $V^*_2$. We argue that this is the optimal strategy as picking anything that is not vertex cover results more than $k$ infected nodes in $V^*_2$. Then we pick arbitrary $k$ nodes as our action in $t=1$ and results a score of $Score(2)=2|v|-k$. The objective function we gain is $5|V|-2k$ if and only if the vertex cover exist. Thus we have proven the ACTS problem to be NP-hard. 
\end{proof}
\section{Algorithm for the ACTS problem}\label{sec:alg}
We introduce \textsc{REMEDY}, an algorithm for selecting nodes to actively screen in the ACTS problem. \textsc{REMEDY}, shown in Algorithm~\ref{alg:BU}, contains two parts: (i) a marginal belief state update that we use for reasoning about the infected status of nodes, and (ii) an algorithm for selecting which nodes to actively screen based on the marginal belief state and an upper bound of the ACTS objective.
\subsection{Belief State Update}\label{modelsection}
Tracking the exact probability that a node is infected in ACTS requires storing $O(2^{|V|})$ values, which is computationally intractable for reasonably sized graphs. Thus, \textsc{REMEDY} maintains a belief state based on the \emph{marginal} probability that each node is infected, requiring only $O(|V|)$ values. In the action choice algorithm, we form an upper bound on the ACTS objective that accounts for the imprecision of the marginal belief state.

% Such independence assumptions have been made in prior literature on the spread of contagion and are widely accepted~\cite{wang2003epidemic,virus-tissec,thres1}. 

The marginal belief update is lines \ref{line:begin}--\ref{line:end} and \ref{line:begin2}--\ref{line:end2} of Algorithm~\ref{alg:BU}. At each time step $t \in \{0,\ldots,T-1\}$, we acquire perfect information about the infected state of any $I$ node when it recovers without active screening with probability $c$. Otherwise its state remains unknown. This naturally recovered node set $C_n(t)$ is given as nodes such that $s(t)=I$ and $s(t+1)=S$. 

Let $x_v(t) \in \left[0,1 \right]$ be a random variable indicating whether node $v$ is in state $I$ at time $t$ and let ${\bf b_v}(t)=[1-x_v(t),x_v(t)]^{\top}$ be the marginal belief vector. For each node, we update an intermediate belief state ${\bf \bar{b}_{v} }(t)=[1-\bar{x}_v(t),\bar{x}_v(t)]^{\top}$ in which $\bar{x}_v(t)=1$ for $v\in C_n(t)$ and $\bar{x}_v(t)=\frac{(1-c)x_v(t)}{(1-x_v(t))+(1-c)x_v(t)}$ for the remaining nodes $v\in V \setminus C_n(t)$. These update steps are in lines \ref{line:begin}--\ref{line:end} of Algorithm~\ref{alg:BU}. This intermediate belief state is then exploited by the action choice subroutine to select $C_a(t)$, the node set we actively cure (line 8). After that, we calculate the marginal belief state of the next time step as ${\bf b_v}(t+1)=\mathbf{B}^N_v(t) {\bf \bar{b}_v}(t)$ and ${\bf b_v}(t+1)=\mathbf{ B}^A_v(t) {\bf \bar{b}_v}(t)$ for $v\in V \setminus (C_n(t) \cup C_a(t))$ and $v\in C_n(t) \cup C_a(t)$ respectively where
\begin{gather}\label{eq1}
\mathbf{B}^N_v(t) = \begin{blockarray}{ccc}
  & \matindex{$S$} & \matindex{$I$}\\
    \begin{block}{c[cc]}
      \matindex{$S$} & 1-p_v & 0 \\
      \matindex{$I$} & p_v & 1 \\
    \end{block}
  \end{blockarray} \;,
  \mathbf{ B}^A_v(t) = \begin{blockarray}{ccc}
  & \matindex{$S$} & \matindex{$I$}\\
    \begin{block}{c[cc]}
      \matindex{$S$} & 1 & 1 \\
      \matindex{$I$} & 0 & 0 \\
    \end{block}
  \end{blockarray} \; 
\end{gather}
and $p_v = 1-\prod_{u\in \delta(v)} (1-\alpha \bar{x}_u(t) )$. %As for $v\in C_n(t) \cup C_a(t)$ we have ${\bf b_v}(t+1)=\mathbf{ B}^A_v(t) {\bf \bar{b}_v}(t)$ where
%\begin{gather}
%\mathbf{ B}^A_v(t) = \begin{blockarray}{ccc}
%  & \matindex{$S$} & \matindex{$I$}\\
%    \begin{block}{c[cc]}
%      \matindex{$S$} & 1 & 1 \\
%      \matindex{$I$} & 0 & 0 \\
%    \end{block}
%  \end{blockarray} \; .
%\end{gather}
These are shown in lines \ref{line:begin2}-\ref{line:end2} of Algorithm~\ref{alg:BU}.
The transition matrix $\mathbf{B}^N$ does not contain parameter $c$ because each node in the $I$ state that did not naturally recover will remain in $I$ state with probability $1$. It is worth noting that, intuitively, to update the marginal belief state for node $v$, one has to calculate the probability of all possible sets of infected neighbors of $v$. However, we show below in Lemma~\ref{prob} that the approach adopted by Eq.~\ref{eq1} to calculate $p_v$ yields the exact probability of $v$ becoming infected by its neighbors given it is currently in $S$, which saves a great amount of computational time.
\begin{lemma} \label{prob}
The exact marginal probabilities of $P[s(t+1)=I|s(t)=S]$ can be calculated by $p_v$ without listing the probability associated with each possible set of infected neighbors.
\end{lemma}
\begin{proof}
The theorem can be proved by induction. For the base case where there is only one neighbor, the probability that node $v$ is infected in the next time step given it is currently in $S$ is $p_{v,1}=\bar{x}_u(1-(1-\alpha)^1)+(1-\bar{x}_u)(1-(1-\alpha)^0)=\alpha\bar{x}_u$. Assume $p_{v,k}=1-\prod\limits_{u\in \delta(v)} (1-\alpha \bar{x}^{t}_{u} )$ for $|\delta(v)\leq k|$ is true, for $|\delta(v)= k+1|$, where $w$ denotes the newly added neighbor, we have:
\begin{align*}
p_{v,k+1}&=p_{v,k}+\bar{x}_{w}\alpha-p_{v,k}\bar{x}_{w}\alpha \\
&=(1-\bar{x}_{w}\alpha)(1-\prod\limits_{u\in \delta(v)\setminus w} (1-\alpha \bar{x}_{u} ))+\bar{x}_{w}\alpha\\
&=1-\prod\limits_{u\in \delta(v)} (1-\alpha \bar{x}_{u} )
\end{align*}
Thus we proved that $p_v$ evaluates the exact probability of $P[s_v(t+1)=I|s_v(t)=S]$.
\end{proof}

\begin{algorithm}[tb]
\caption{\textsc{REMEDY}}
\label{alg:BU}
\textbf{Input}: $\mathbf{ A}$, ${\bf b}(t)$, $\alpha$, $c$, $C_n(t), t, T, k$\\
\textbf{Output}: $C_a(t)$, ${\bf b}(t+1)$
\begin{algorithmic}[1] %[1] enables line numbers
\FOR{$v \in V$} \label{line:begin}
\IF {$v \in C_n(t)$}
\STATE ${\bf \bar{b}_v}(t)\gets \left[0,1\right]^{\top}$
\ELSE
\STATE ${\bf \bar{b}_v}(t)\gets \frac{\left[(1-x_v(t),(1-c)x(t)\right]^{\top}}{((1-x_v(t))+(1-c)x_v(t))}$
\ENDIF
\ENDFOR \label{line:end}

\STATE $C_a(t) \gets \textsc{ActionChoice}(\mathbf{ A}, {\bf \bar{b}}(t), \alpha, c, C_n(t), t, T, k)$ \label{line:action}

\FOR{$v \in V$}\label{line:begin2}
\IF {$v \in V \setminus C_n(t)\cup C_a(t)$}
\STATE ${\bf b_v}(t+1)\gets \mathbf{ B}^N_v(t) {\bf \bar{b}_v}(t)$
\ELSE
\STATE ${\bf b_v}(t+1) \gets \mathbf{ B}^A_v(t) {\bf \bar{b}_v}(t)$
\ENDIF
\ENDFOR\label{line:end2}
\STATE \textbf{return}  $C_a(t)$, ${\bf b}(t+1)$
\end{algorithmic}
\end{algorithm}

\subsection{Action Choice Algorithm}
We now turn our attention to selecting the set of nodes to actively screen, i.e., line~\ref{line:action} in Algorithm~\ref{alg:BU}.
One fast yet naive approach to this problem is to select the node set with maximum marginal belief to be in $I$ state. This approach can be in $O(|V|)$, but it does not take the network structure and future infection probabilities into account. Another approach is to choose nodes to cure in order to minimize the largest eigenvalue of the network that results from deleting or permanently actively screen the same set of nodes~\cite{thres1}. This approach guarantees that the infection is eradicated in the long term if the largest eigenvalue can be reduced below $\frac{c}{\alpha}$ for sufficient budget $k$. However, it does not take the belief state into account, nor does it consider how many people become infected before the disease is eradicated. These methods are examined in the experiments as our baselines.
We seek to do better by minimizing an upper bound of the  ACTS objective directly.

%Remark that though we store the marginal belief state, the upper bound is taken w.r.t.\ to the \emph{true} ACTS objective. 
We develop two different algorithms for action choice: \textsc{Full-ActionChoice} that looks ahead through all future actions and \textsc{Fast-ActionChoice}, a less computationally intensive variant, that considers only the current action, allowing it to exploit eigenvalue decomposition. We refer to \textsc{REMEDY} with \textsc{Full-ActionChoice} as \textsc{Full-REMEDY} and \textsc{REMEDY} with \textsc{Fast-ActionChoice} as \textsc{Fast-REMEDY}. Noted that in both \textsc{Full-REMEDY} and \textsc{Fast-REMEDY}, we change the action based on the observation $C_n(t)$ in each round.

We begin by deriving an upper bound for the ACTS objective starting with some preliminary notation. To encapsulate the effect of active-screening toward our objective function, we define the $|V| \times |V|$ diagonal action matrix $\mathbf{R}_{a}(t)$ at time $t$ as $\mathbf{R}_{a}(t)_{v,v}=1$ if and only if $v \in C_a(t)$ and $0$ otherwise. For the current round, say $t_0$, we observe the nodes that are cured and need to decide the nodes to actively screen. We define the \emph{naturally cured matrix} $\mathbf{R}_{n}(t_0)$ as $\mathbf{R}_{n}(t_0)_{v,v}=1$ if and only if $v \in C_n(t_0)$, which encapsulates the knowledge we gain from natural recovery in the current round.
Let vector ${\bf x}(t)$ represent $x_v(t)$ for all $v$. To bound ${\bf x}(t)$ across all time steps given the actions we take, let $\mathbf{M}'=\alpha \mathbf{A}+\mathbf{I}$, where $\mathbf{A}$ is the adjacency matrix and $\mathbf{I}$ is the identity matrix, define the  \emph{upper bound transition matrix} for the current round as $\mathbf{M}_a(t_0)=(\mathbf{I}-\mathbf{R}_a(t_0)-\mathbf{R}_n(t_0))\mathbf{M}'$, and as $\mathbf{M}_a(t)=(\mathbf{I}-\mathbf{R}_{a}(t))\mathbf{M}$ for future rounds $t > t_0$, where $\mathbf{M}=\alpha \mathbf{A}+(1-c)\mathbf{I}$.
\begin{theorem}  \label{upperbound}
Let the current time be $t_0$. $\mathbf{M}_a$ is defined as above for $t_0$ and $t > t_0$. The ACTS objective (Eq.~\ref{eq:obj}) that the actions affect is bounded above as follows:
\begin{align}
\mathbbm{E}[\sum_{t=t_0}^{T}\sum_{v\in V} |s_v(t)=I|] \leq F(\mathbf{R}_a(t_0),...,\mathbf{R}_a(T)),\\
\mbox{where } F=\mathbbm{1}^{\top}\sum_{t=t_0}^{T} \prod^{t}_{\tau=t_0} \mathbf{M}_a(\tau) {\bf x}(t_0)\\
\mbox{and } \prod^{t}_{\tau=t_0} \mathbf{M}_a(\tau)=\mathbf{M}_a(t)\mathbf{M}_a(t-1)...\mathbf{M}_a(t_0).
\end{align}
\end{theorem}
\begin{proof}
\sloppy
Observe that the conditional probability $P\left[{\bf s_v}(t+1)=I|{\bf s_v}(t)=S\right]$ is bounded by
\begin{align*}\label{jenson}
%\mathbbm{E}[x_u|u\in \delta(v)]1-(1-\alpha)^{|\delta(v)|} \leq \notag\\
P\left[{\bf s_v}(t+1)=I| {\bf s_v}(t)=S\right]\leq
1-(1-\alpha)^{\sum_{u \in \delta(v)} {x_u}}.
\end{align*}
%where $\mathbbm{E}[x_u|u\in \delta(v)]=\frac{1}{|\delta(v)|}\sum_{u \in \delta(v)} {x_u}$ denotes the average probability of being in state $I$ of the neighbors of $v$.
We show this as follows: $P\left[{\bf s_v}(t+1)=S| {\bf s_v}(t)=S\right] = \sum_{m=0}^{|\delta(v)|} p_m (1 - \alpha)^m$ where $m$ denotes number of infected neighbors of $v$ and $p_m$ is probability of $m$ neighbors infected. Then, $\sum_{m=0}^{|\delta(v)|} p_m (1 - \alpha)^m = \mathbbm{E}[(1 - \alpha)^m] \geq (1 - \alpha)^{\mathbbm{E}[m]} = (1-\alpha)^{\sum_{u \in \delta(v)} {x_u}}$ yields the result by applying Jensen's inequality.
%The upper bound of the above equation is the extreme case where there are always $|\delta(v)|\mathbbm{E}[x_u|u\in \delta(v)]$ of neighbors that are infected, where $\mathbbm{E}[x_u|u\in \delta(v)]=\frac{1}{|\delta(v)|}\sum_{u \in \delta(v)} {x_u}$ denotes the average probability of being in state $I$ of the neighbors of $v$. 
%The lower bound is derived from the extreme case where $\mathbbm{E}[x_u|u\in \delta(v)]$ of the time all the neighbors are infected and other times no neighbors are infected.
We approximate the right hand side with a first-order Taylor series expansion as $\alpha \sum_{u \in \delta(v)}x_u(t)$, yielding
\begin{align*}
x_v(t+1)&\leq (1-x_v(t)) \alpha \sum_{u \in \delta(v)}x_u(t)+ x_v(t)(1-c).
\end{align*}
Using a vector ${\bf x}(t)$ to represent $x_v(t)$ for all $v$, the above yields the following inequality in vector form:
%\begin{align*}
%x(t+1)&=\alpha A X(t)(\mathbbm{1}-x(t))+(1-c)x(t)\\
%&=\alpha A (x(t)-X(t)x(t)))+(1-c)x(t)\\
%&= M x(t)- \alpha A X(t)x(t)
%\end{align*}
\begin{align}
{\bf x}(t+1)\leq \mathbf{M} {\bf x}(t)- diag(\alpha \mathbf{A} {\bf x}(t)){\bf x}(t),
\end{align}
where $\mathbf{M}=\alpha \mathbf{A}+(1-c)\mathbf{I}$ and $\mathbf{A}$ is the adjacency matrix. 
%The largest eigenvalue of $\mathbf{A}$, $\lambda_A$, has the following relationship with $\mathbf{M}$'s largest eigenvalue: $\lambda_M=\alpha \lambda_A+(1-c)$. 
We drop the negative term $\textrm{diag}(\alpha \mathbf{A} {\bf x}(t)){\bf x}(t)$, and use $\mathbf{M}{\bf x}(t)$ as an upper bound. 

While the above holds without intervention, we need to use the transition matrix with given interventions $R_a(t_0), \ldots, R_a(T)$ and knowledge of $C_n(t_0)$. We show in the appendix that this matrix is precisely $M_a(t_0)$ for the current time $t_0$ and $M_a(t)$ for $t > t_0$. Then, ${\bf x}(t+1)\leq  \mathbf{M}_a(t) {\bf x}(t)$ for all $t \geq t_0$ and $\mathbbm{E}[\sum_{t=t_0}^{T}\sum_{v\in V} |s_v(t)=I|] =\sum_{t=t_0}^{T}\mathbbm{1}^{\top}{\bf x}(t) $ yields the desired result.
\end{proof}

Given that the function $F$ upper bounds our objective function, we next describe the method we use to select the action matrix $R_a(t)$ that minimize $F$ for every time step. Distinct from previous literature, our objective takes into account the number of infected nodes at each time step. We also have the flexibility to change the action we take based on the observation we make in each round. Such flexibility results in an even larger solution space,  $\binom{n}{k}^T$. Since our problem is NP-hard, we apply a Frank-Wolfe style method to attempt to optimize $F$~\cite{frank1956algorithm}. 

Frank-Wolfe is a gradient-base algorithm runs for some number $L$ steps by starting with some arbitrary feasible points and updates it in two steps: (i)computes the gradient of the objective at the current point (ii)find the point which optimizes the gradient over the feasible set and step toward it. For the first step, we need the gradient of $F$ w.r.t.\ the available action choices. We relax our optimization to a continuous problem by allowing $\mathbf{R}_a(t)_{v,v}$ to take real values between $0$ and $1$ (instead of binary $0,1$), which can be interpret as the probability of choosing node $v$. As a consequence, the feasible solution space is the convex hull of the binary $\mathbf{R}_a(t)$. We denote this convex hull $\Psi$. By taking the derivative of $F$, the gradient w.r.t. action at each time $t$ is
\begin{equation}
\label{gradeq}
%\frac{\partial F}{\partial \mathbf{R}_a(t)}=-\sum_{t'=t+1}^{T}(\prod^{t'}_{\tau=t+1} \mathbf{M}_{a}(\tau))^{\top}\mathbbm{1} {\bf x}^{\top}(t_0) (\prod^{t-1}_{\tau=t_0} \mathbf{M}_{a}(\tau))^{\top},
\frac{\partial F}{\partial \mathbf{R}_a(t)}=-\sum_{t'=t+1}^{T}\prod^{t+1}_{\tau=t'} \mathbf{M}^{\top}_{a}(\tau)\mathbbm{1} {\bf x}^{\top}(t_0) \prod^{t_0}_{\tau=t-1} \mathbf{M}^{\top}_{a}(\tau),
\end{equation}
The above gradient is a matrix $\Delta(t)$, where the diagonal elements $\Delta(t)_{v,v}$ represents the gradient w.r.t. choice of node $v$ to actively screen at time $t$. Given the gradient and current information, an approximately optimal action for all times in the continuous relaxation can be obtained through a projected gradient descent or a Frank-Wolfe style algorithm, yielding \textsc{Full-ActionChoice} (Algorithm~\ref{alg:FullR}).
\begin{algorithm}[tb]
\caption{\textsc{Full-ActionChoice}}
\label{alg:FullR}
\textbf{Input}: $\mathbf{ A}$, ${{\bf \bar{b}}(t_0)}$, $\alpha$, $c$, $T$, $t_0$, $k$\\
\textbf{Output}: $C_a(t_0)$
\begin{algorithmic}[1] %[1] enables line numbers
\STATE $\mathbf{R}_a^0(t) \gets \mathbf{0} \quad \forall t$
\FOR{$l =1... L$}
\FOR{$t =t_0... T$}
\STATE $\mathbf{\Delta}(t) \gets \textsc{GradientOracle}(\mathbf{R}_a^{l-1})$ \label{gradline}
\STATE $\mathbf{R}_a^*(t) \gets \textsc{ProjectFeasible}(\Delta, k)$ \label{proj}
\STATE $\mathbf{R}_a^l(t) \gets \gamma_l \mathbf{R}_a^{l-1}(t)+(1-\gamma_l) \mathbf{R}_a^*(t)$ \label{step}
\ENDFOR
\ENDFOR
\STATE $C_a(t_0) \gets \arg\max_{k} \mathbf{R}^L_a(t_0)$ 
\STATE \textbf{return} $C_a(t_0)$
\end{algorithmic}
\end{algorithm}

In the \textsc{Full-ActionChoice} algorithm, we first set an arbitrary feasible point in $\Psi$ for each time step, say $\mathbf{R}^0_a(t)=\mathbf{0}$ for iteration $0$. 
%We then update the candidate solution for every time step simultaneously in each iteration $l$ in three steps. 
In each iteration, first we calculate the gradient for the current feasible point using Eq.~\ref{gradeq} as our \textsc{GradientOracle} in line \ref{gradline}. Second, we project the resultant point toward the current best solution on the gradient hyper-plane for every time step simultaneously. We do so simply by greedily selecting $k$ nodes with largest $\Delta(t)_{v,v}$ as our current best solution $\mathbf{R}^{*}_a(t)$ in line~\ref{proj}. Third, we set the initial point $\mathbf{R}^l_a(t)$ of the next iteration in line~\ref{step}, in which $\gamma_{l}=2/(l+2)$ is the step size of Frank-Wolfe algorithm. Since $\Psi$ is convex and $\mathbf{R}^l_a(t)$ is the convex combination of two feasible points, it is guaranteed that it will remain in the convex hull $\Psi$ by the update. After the iteration finishes, we output our action in the current round by greedily selecting $k$ nodes of the relaxed $\mathbf{R}^L_a(t_0)$ of the final iteration and wait for new information from the next round to arrive.

The \textsc{Full-REMEDY} algorithm considers future actions simultaneously and has time complexity of $O(T^2|V|^{\omega})$, where $\omega$ an exponent arising from matrix multiplication. However, calculating such solutions for a very large network---which is often the case for contact tracing---could be time consuming for low resource regions. To reduce time complexity, we further simplify the upper-bound function by assuming that no actions are taken in the future rounds and ignore their effect on the current decision making to derive  \textsc{Fast-ActionChoice} (Algorithm~\ref{alg:DE}). By ignoring future actions, the action matrix $M_a(t)$ in \textsc{Full-REMEDY} is simplified to constant $M$. The contribution of actively screening each node can be written as the following vector form:
\begin{align}
\mathbbm{1}^{\top}\sum\nolimits_{\tau=0}^{T-t_0-1}\mathbf{M}^{\tau}diag(\mathbf{M}_n {\bf x}(t_0)),
\end{align}
where $\mathbf{M}_{n}=(\mathbf{I}-\mathbf{R}_n(t_0))\mathbf{M}'$. Now, since $\mathbf{M}$ is the same for every future round, $\mathbf{ M}$ can be decomposed as $\mathbf{Q} \mathbf{\Lambda} \mathbf{Q}^{\top}$ ahead of time, where $\mathbf{Q}$ is a matrix comprised of the eigenvectors of $\mathbf{M}$, and $\mathbf{\Lambda}$ a diagonal matrix comprised of the eigenvalues along the diagonal. Such a matrix can be approximated by calculating only the top $m$ largest eigenvalues and their eigenvectors using the Lanczos algorithm (\cite{lanczos1950iteration}) that has a complexity of $O(|E|)$ (assuming the large network is sparse), yielding the \textsc{Fast-ActionChoice} shown in Algorithm \ref{alg:DE}. The approximate $\mathbf{M}$ is given by $\mathbf{Q}_m\mathbf{\Lambda}_m \mathbf{Q}_m^{\top}$, where these matrices are computed in line 3. In line 5, the well-known result $(\mathbf{Q}_m\mathbf{\Lambda}_m \mathbf{Q}_m^{\top})^{\tau} = \mathbf{Q}_m\mathbf{\Lambda}_m^{\tau} \mathbf{Q}_m^{\top}$ is used to approximate $\mathbf{M}^{\tau}$. The time complexity of \textsc{Fast-REMEDY} is $O(|V|^2)$ assuming constant $m$.
\begin{algorithm}[tb]
\caption{\textsc{Fast-ActionChoice}}
\label{alg:DE}
\textbf{Input}: $\mathbf{A}$, ${{\bf \bar{b}}(t_0)}$, $\alpha$, $c$, $T$, $t_0$, $k$\\
\textbf{Output}: $C_a(t_0)$
\begin{algorithmic}[1] %[1] enables line numbers
\IF {$t_0=0$}
\STATE $\mathbf{M} \gets \alpha \mathbf{A}+(1-c)\mathbf{I}$
\STATE $\mathbf{Q}_m,  \mathbf{\Lambda}_m \gets \textsc{Lanczos}(\mathbf{M},m)$
\ENDIF
\STATE $\textbf{Scores} \gets \mathbbm{1}^{\top}\mathbf{Q}_m(\sum_{\tau=0}^{T-t_0-1}\mathbf{\Lambda}_m^{\tau}) \mathbf{Q}_m^{\top} diag(\mathbf{M}_n {\bf x}(t_0))$
\STATE $C_a(t) \gets$ $k$ nodes with highest scores in vector $\textbf{Scores}$
\end{algorithmic}
\end{algorithm}

\section{Experiments}\label{sec:experiments}
We perform experiments comparing \textsc{Fast-} and \textsc{Full-REMEDY} to baselines on a variety of real-world, publicly available datasets. Table~\ref{resultstable} lists all the networks and their properties. Most of the networks were collected in actual human contact settings. The networks used here have varied sizes ($|V|$), average degrees ($d$), assortativities ($\rho_D$), and epidemic thresholds ($1/\lambda_A$).

\begin{table*}
\centering
\scalebox{0.78}{
 \begin{tabular}{l c c c c|cccccc} 
 \toprule
 \multirow{2}{*}{Network} & \multirow{2}{*}{$|V|$} & \multirow{2}{*}{$\frac{1}{\lambda_A^*}$} & \multirow{2}{*}{$d$} & \multirow{2}{*}{$\rho_D$} & \multicolumn{6}{|c}{$\alpha=0.1, c=0.1$} \\ \cmidrule{6-11}
 &  &  & & & Random & Max-Degree & Eigenvalue  & Max-Belief  & Fast-REMEDY & Full-REMEDY \\
 \midrule\midrule
 \textbf{Hospital}~\cite{hospital} & 75 & 0.027 & 15.19 & -0.18 & 144 & 150 & 151  & 150 & \textbf{156} & \textbf{160} \\
 \textbf{India}~\cite{jackson} & 202 & 0.095 & 3.43 & 0.02 & 605 & 470 & 420  & 636 & \textbf{890} & \textbf{901} \\ 
 \textbf{Face-to-face}~\cite{infectious} & 410 & 0.042 & 6.74 & 0.23 & 809 & 843 & 745 &1057& \textbf{1297} & \textbf{1409} \\
 \textbf{Flu}~\cite{stanfordflu} & 788 & 0.003 & 150.12 & 0.05 & 1336 & 1421 & 1431 & 1438 & \textbf{1443} & \textbf{1446}\\
 \textbf{Irvine}~\cite{panzarasa2009patterns} & 1899 & 0.021 & 7.29 & -0.18 & 4630 & 5741 & 3692 &4957& \textbf{6676} & \textbf{7821} \\
 \textbf{Escorts}~\cite{escort} & 16730 & 0.032 & 2.33 & -0.03 & 27400 & 30167 & TLE  & 29493 & \textbf{46549} & TLE\\
 \bottomrule
 \end{tabular}}
\caption{\textbf{Improvement} of the objective function over \textbf{None} (The larger the better). TLE implies time limit of 24 hours for all rounds exceeded. }
\label{resultstable}
\end{table*}

\vspace*{2mm}
\noindent\textbf{Setting.}\hspace{2mm}In all simulations, we assume the budget $k$ allows for screening and treatment of $20\%$ of the total population $|V|=n$ per round. All results are averages over $30$ simulation runs.

\vspace*{2mm}
\noindent\textbf{Setup.}\hspace{2mm}In the real world, active screening is performed only after conducting initial surveys on the prevalence and incidence of the disease. To simulate this, we run our experiments in two stages. 

\vspace*{2mm}
\noindent\textbf{Stage 1 (Survey Stage).}\hspace{2mm}This stage starts at $t=0$ with $25\%$ of individuals in $I$ and ends at $t=10$. No active screening is done and the disease evolves naturally. The initial belief $\bf b(0)$ for all nodes is assumed to be $[0.5, 0.5]^{\top}$ since we have no prior information. However beliefs are updated when individuals come to the clinic voluntarily (with probability $c$). This belief update requires knowledge of $\alpha$ and c. There is a rich literature of how to estimate the disease parameters ($\alpha$ and $c$) in this stage and these methods have been tested on real world scenarios~\cite{kirkeby2017methods,saad2016mathematical,dong2012modeling}. Here, we assume that such parameters are known. 

Such parameters can vary from disease to disease. For example, the transmission rate of Pertussis can be as high as $0.47$ for certain age groups in \cite{hethcote1997age} and as low as $0.035$ for Syphilis \cite{saad2016mathematical}. The cure rate also depends on how resourceful are the target regions. We fix the value of our parameter to a reasonable value $(\alpha, c)=(0.1,0.1)$ first for comparison and evaluate a wide range of values afterwards.

\vspace*{2mm}
\noindent\textbf{Stage 2 (ACTS Stage).}\hspace{2mm}Here, we consider various screening algorithms. We perform active screening from $t=11$ to $t=T=20$ to represent 5 years of time (each round is 6 months~\cite{cdc2}). Beliefs are updated according to the belief update scheme presented in Section~\ref{modelsection}.

\subsection{Metrics} In Table~\ref{resultstable}, we compare the outcomes of these screening strategies compared to no intervention (\textbf{None}) based on the total infected number over time. In \textbf{None}, the evolution of the health states is based on disease dynamics with no active screening for all $T$ timesteps. 

\vspace*{2mm}
\noindent\textbf{Comparison with Baselines.}\hspace{2mm}Given the lack of previous algorithms for our problem setting, we measure the performance of \textsc{REMEDY} against baselines: 
\begin{itemize}
\denselist
\item[(1a)] \textbf{Random}: Randomly select nodes for active screening. 
\item[(1b)] \textbf{MaxDegree}: Greedily choose nodes of the largest degree until the budget is reached. %This baseline uses only the graph structure information and thus does not update the belief state.
\item[(1c)] \textbf{Eigen}: Greedily choose nodes that reduce the largest eigenvalue of $A$ the most until the budget is reached. %This baseline uses only the graph structure information.
%\item[(1d)] \textbf{DynamicEigen}: After multiplying $\bar{x}[v]$ to both column and row (thus the modified matrix remains symmetry), greedily choose nodes that reduce the largest eigenvalue the most. This baseline uses both belief state and the graph structure \cite{bhattacharyarepeated}.
\item[(1d)] \textbf{MaxBelief}: Greedily choose nodes with maximum probability of being in the $I$ state. 
\end{itemize}

We test these algorithms on the following realistic contact networks collected from a diversity of sources and methods. The networks are carefully selected to have rich variety of densities, structures and sizes (ranging from $75$ to $16730$ nodes). 
\begin{itemize}
\denselist
\item[(2a)] \textbf{Hospital}~\cite{hospital}: A contact network collected in a university hospital in order to study path of disease spread. 
\item[(2b)] \textbf{India}~\cite{jackson}: A human contact network collected from a rural village in India where active screening with limited budget may take place.
%\item[(2c)] \textbf{Face-to-face}~\cite{infectious}: A network describing face-to-face behavior during the exhibition INFECTIOUS: STAY AWAY in 2009 at the Science Gallery in Dublin, in which influenza might spread through close contact of individuals.
\item[(2c)] \textbf{Face-to-face}~\cite{infectious}: A network describing face-to-face contact in which influenza might spread through the close contact of individuals.
\item[(2d)] \textbf{Flu}~\cite{stanfordflu}: This network captures close proximity interactions in an American high school. The network is highly-dense ($\lambda_A>300$) with small-world properties and a relatively homogeneous distribution. 
%\item[(2e)] \textbf{Irvine}~\cite{panzarasa2009patterns}: A friendship network collected from students in University of California Irvine. It is a representative network used to study rumor modeled as epidemic spread.
\item[(2e)] \textbf{Irvine}~\cite{panzarasa2009patterns}: A friendship network, which is a representative network used to study rumors modeled as epidemic spread.
\item[(2f)] \textbf{Escort}~\cite{escort}: A large sexual contacts between escorts and sex buyers collected for a six-year period, in which STD may be spread.
\end{itemize}
In Table~\ref{resultstable} higher numbers indicate a larger improvement against \textbf{None}. In most cases, both versions of \textsc{REMEDY} make substantial improvements over all baselines, and as expected \textsc{Full-REMEDY} has better solution quality than \textsc{Fast-REMEDY}. A typical number of infected node number over time result is shown in Fig.~\ref{evolve}.  In the \textbf{Hospital} and \textbf{Flu} network, due to the network size or homogeneity, respectively, it is difficult for any algorithm to provide any improvement from a random intervention. In the \textbf{Escort} network, which is the largest network analyzed, we assume that the size of the network makes any algorithm slower than $O(|V|^2)$ impractical.
However, \textsc{Fast-REMEDY} continues to perform better than the baselines that could be run. We examined the performance of \textsc{REMEDY} algorithm on a variety range of $\alpha$ and $c$. \textsc{Fast-} and 
\textsc{Full-REMEDY} continue to perform better than their closest competitor. (see Fig.~\ref{vary}, which shows selected scenarios).
%\subsection{Impact of Variations in Disease Parameters}
%In this section we analyze the impact of changing the disease parameters by varying $\alpha$ and $c$. Due to space constraints, we only demonstrate performance of Random, Fast-REMEDY and Full-REMEDY to illustrate the impact of variation in disease parameters. Fig.~\ref{vary} shows the results.
\begin{figure}[ht]
\includegraphics[width=0.45\textwidth,keepaspectratio]{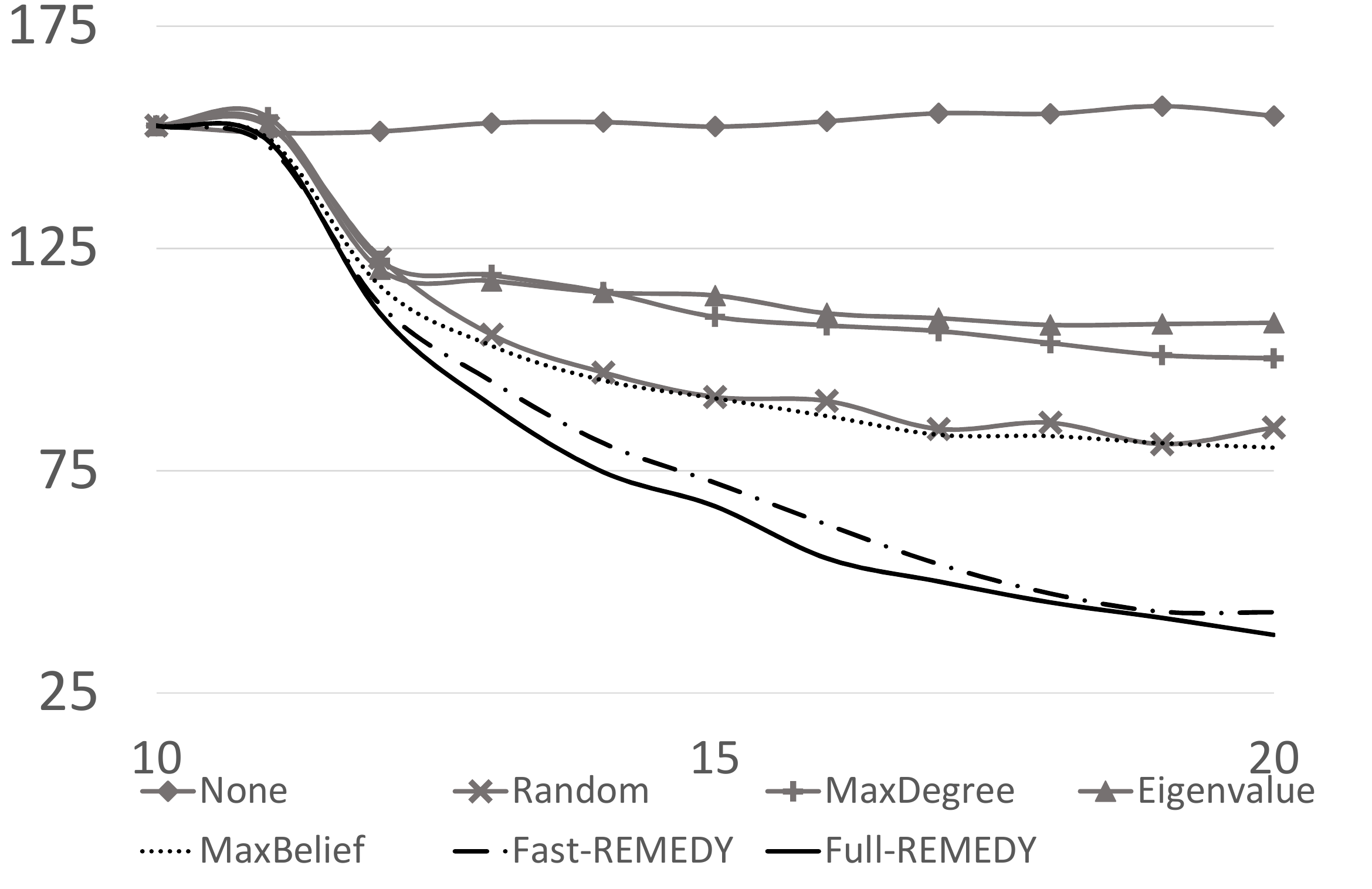}%
\centering
\caption{Number of infected nodes vs.\ time in \textbf{India} network.}\label{evolve}
\end{figure}
\begin{figure}[ht]
\centering
\null\hfill
\subfloat[c=0.1]{%
  \includegraphics[width=0.23\textwidth,keepaspectratio]{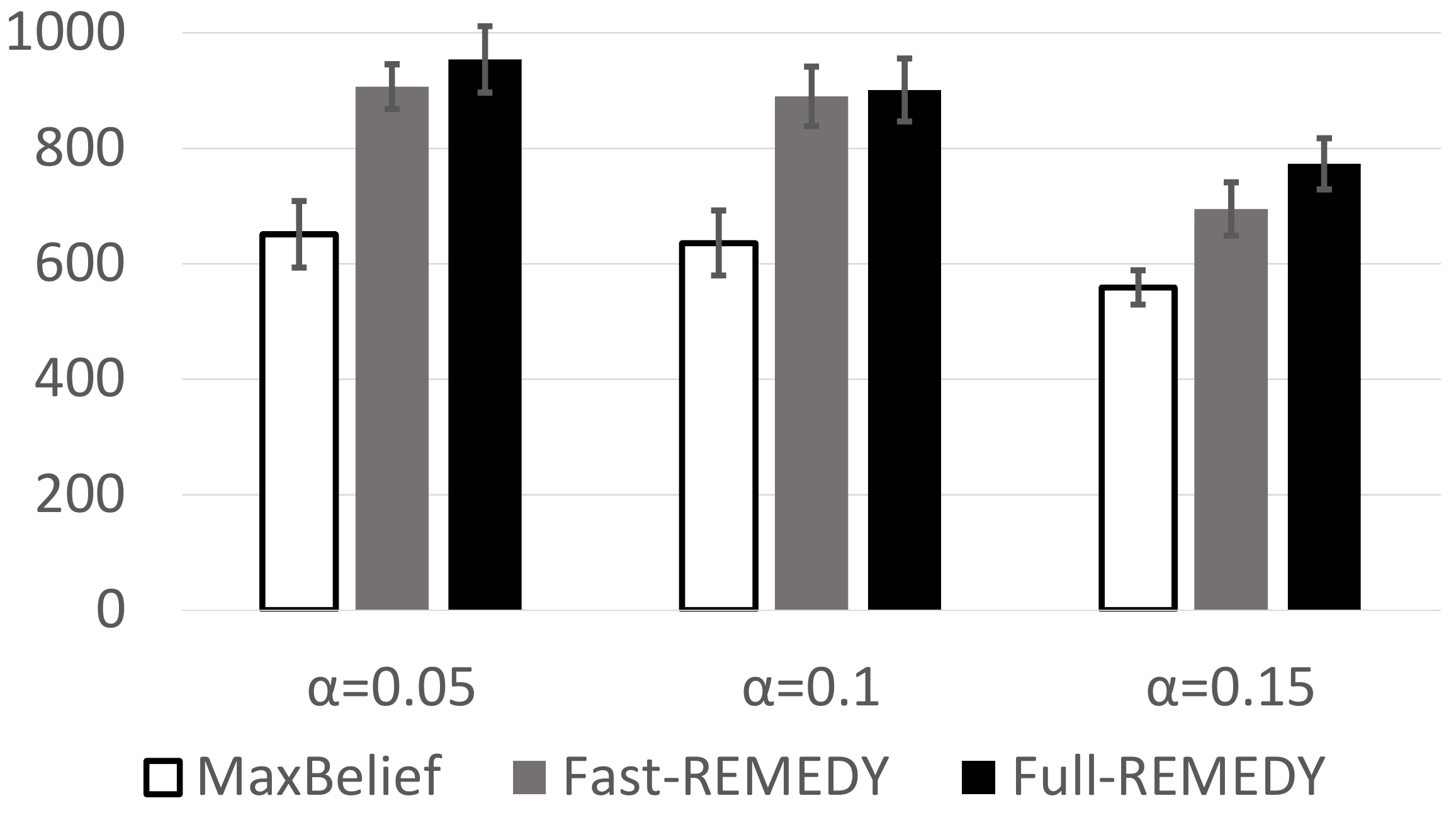}%
}\hfill
\subfloat[a=0.1]{%
  \includegraphics[width=0.23\textwidth,keepaspectratio]{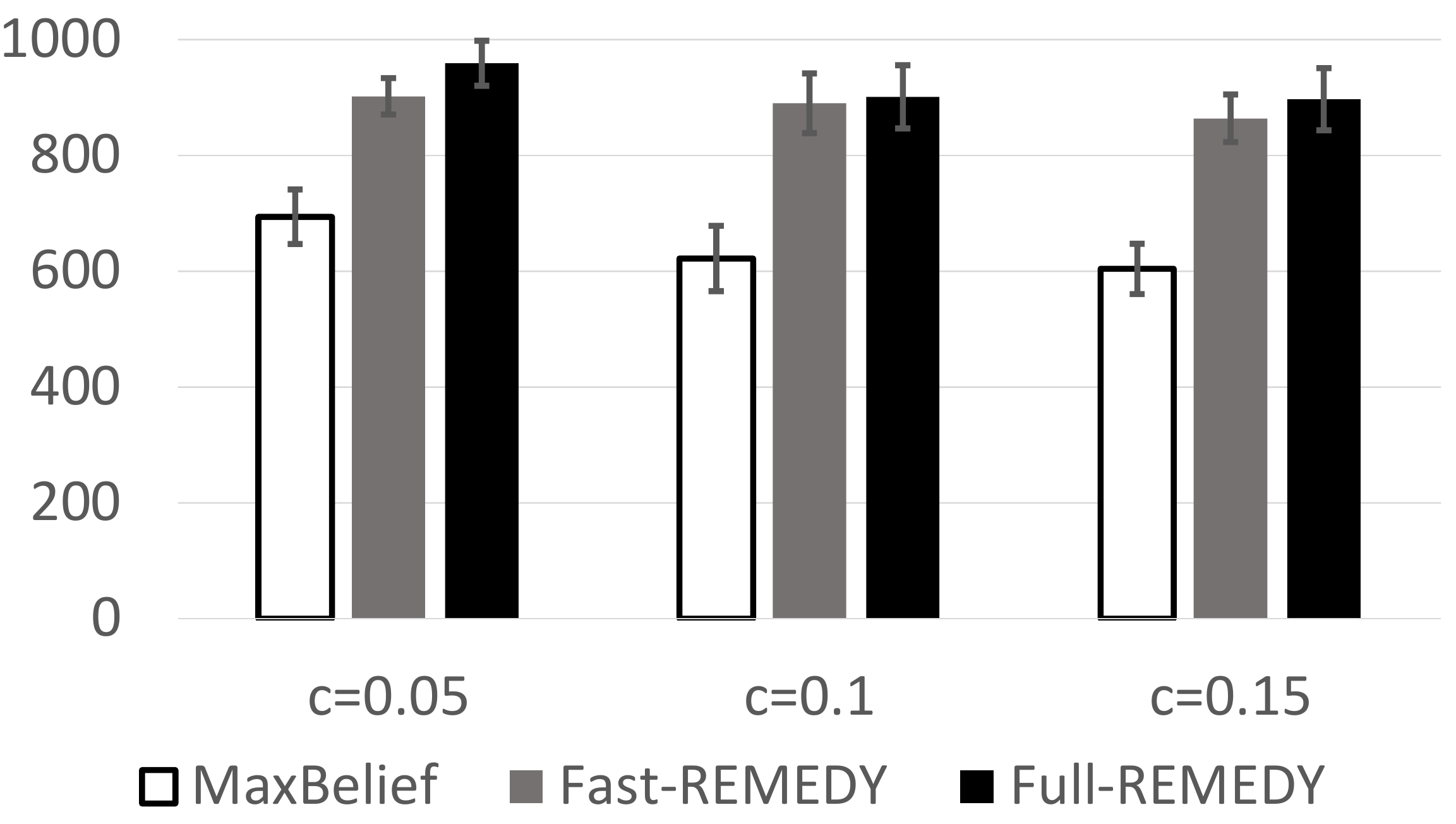}%
}\hfill\null
\centering
\caption{\textbf{Improvement} over \textbf{None} under different parameter sets for \textbf{India} network. Both \textsc{Full-REMEDY} and \textsc{Fast-REMEDY} outperform MaxBelief. \textsc{Full-REMEDY} outperforms \textsc{Fast-REMEDY} especially for larger $\alpha$.}\label{vary}
\end{figure}
%The first phenomenon worth noting is that in Fig.~\ref{vary} (a), the more severe the epidemic (higher $\alpha$), the larger the improvement between Fast-REMEDY with Full-REMEDY. The second phenomenon worth noticing is that the performance of REMEDY algorithms are more sensitive to transmission rate($\alpha$) change then nature cure rate($c$) change. Next, we conduct case studies for various budget on fixed realistic disease parameter set of the corresponding networks.
\subsection{Impact of Budget}
Determining the improvement an intervention can achieve with various budgets is critical when informing health policy. We therefore find the improvement possible over different budget values for two realistically modeled diseases: influenza and syphilis.

\vspace*{2mm}
\noindent\textbf{Influenza.}\hspace{2mm}For influenza, we use the parameters that previous literature estimated through a continuous survey administered in a student residence hall community \cite{dong2012modeling}. The transition rate is estimated to be $\alpha=0.024$ per infectious neighbor and the self cure rate is estimated to be $c=0.3$. We test the algorithms on the \textbf{Face-to-face} network, since this network is used to study the dynamics of SIS-type epidemic spread in its original paper \cite{infectious}. 

Fig.~\ref{real} (a) shows that both \textsc{Fast-REMEDY} and \textsc{Full-REMEDY} outperform other baselines under realistic settings. The difference between algorithms grows larger as the budget increases. According to Prakash et al.~\cite{thres1}, such a network requires at least $k/n\geq \alpha \lambda_A=57\%$ for random screening to fully eradicate the disease. However, the epidemic dies out at the end of 20th round when \textsc{Full-REMEDY} is deployed with a budge of only $k/n=15\%$. 

\vspace*{2mm}
\noindent\textbf{Syphilis.}\hspace{2mm}Saad-Roy et al.~\cite{saad2016mathematical} provides empirically-derived syphilis parameters for our model. The natural cure rate is estimated to be $c=0.01$ and transmission rate $\alpha=0.035$. The network used here is the \textbf{Escort} network with $16730$ nodes, which is a STD contact network. Because the network is large, we show only the algorithms that does not exceed running time due to time complexity, which are Random, Max-Degree, Max-Belief and \textsc{Fast-REMEDY}. Fig.\ref{real} (b) shows that \textsc{Fast-REMEDY} achieves significantly better results than all other baselines. On average, it saves 1140, 2900, and 4600 people from becoming infected every six months for $5\%$, $10\%$ and $15\%$ budgets, respectively.
\begin{figure}[ht]
\centering
\null\hfill
\subfloat[Influenza]{%
  \includegraphics[width=0.26\textwidth,keepaspectratio]{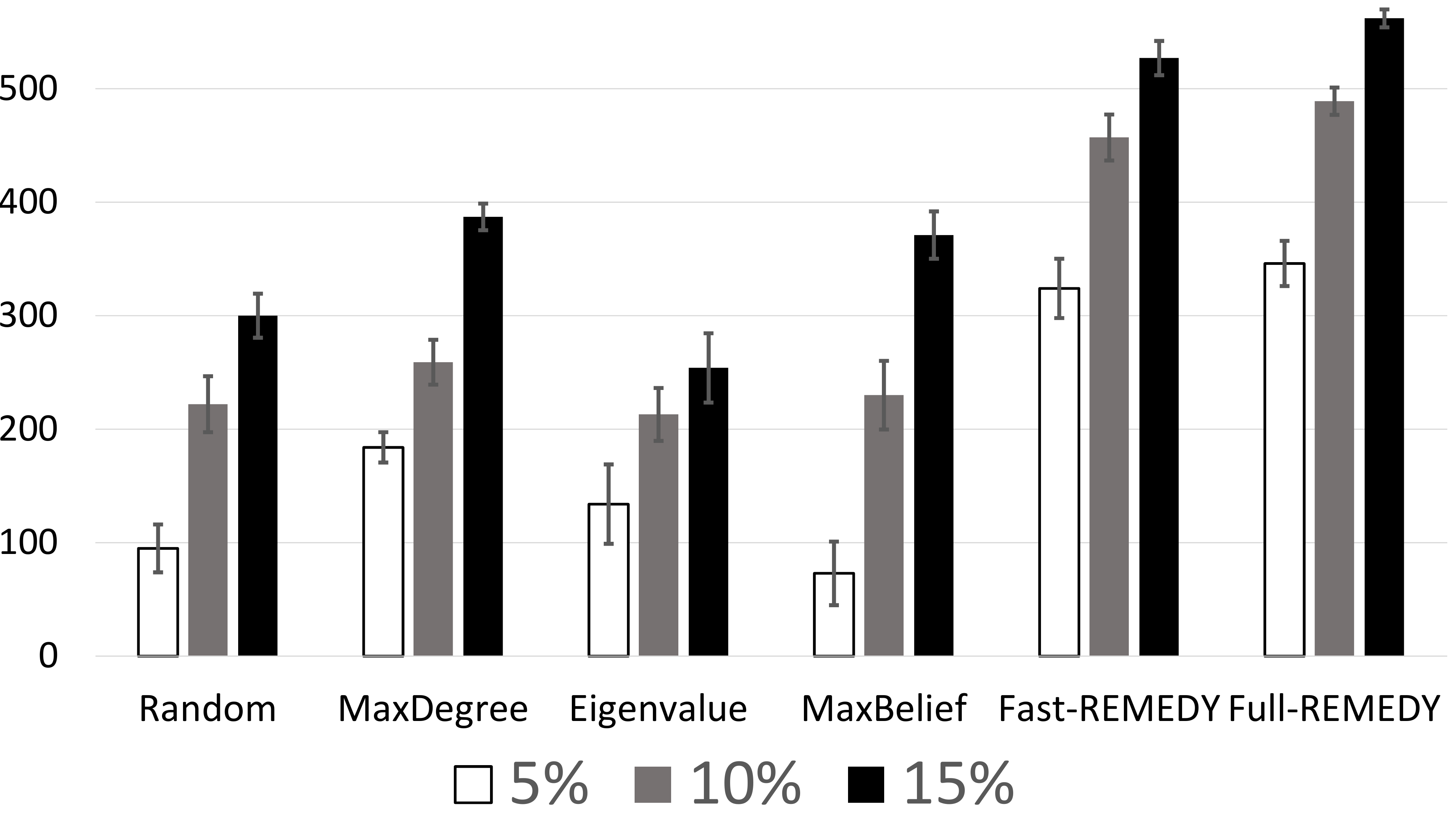}%
}\hfill
\subfloat[Syphilis]{%
  \includegraphics[width=0.20\textwidth,keepaspectratio]{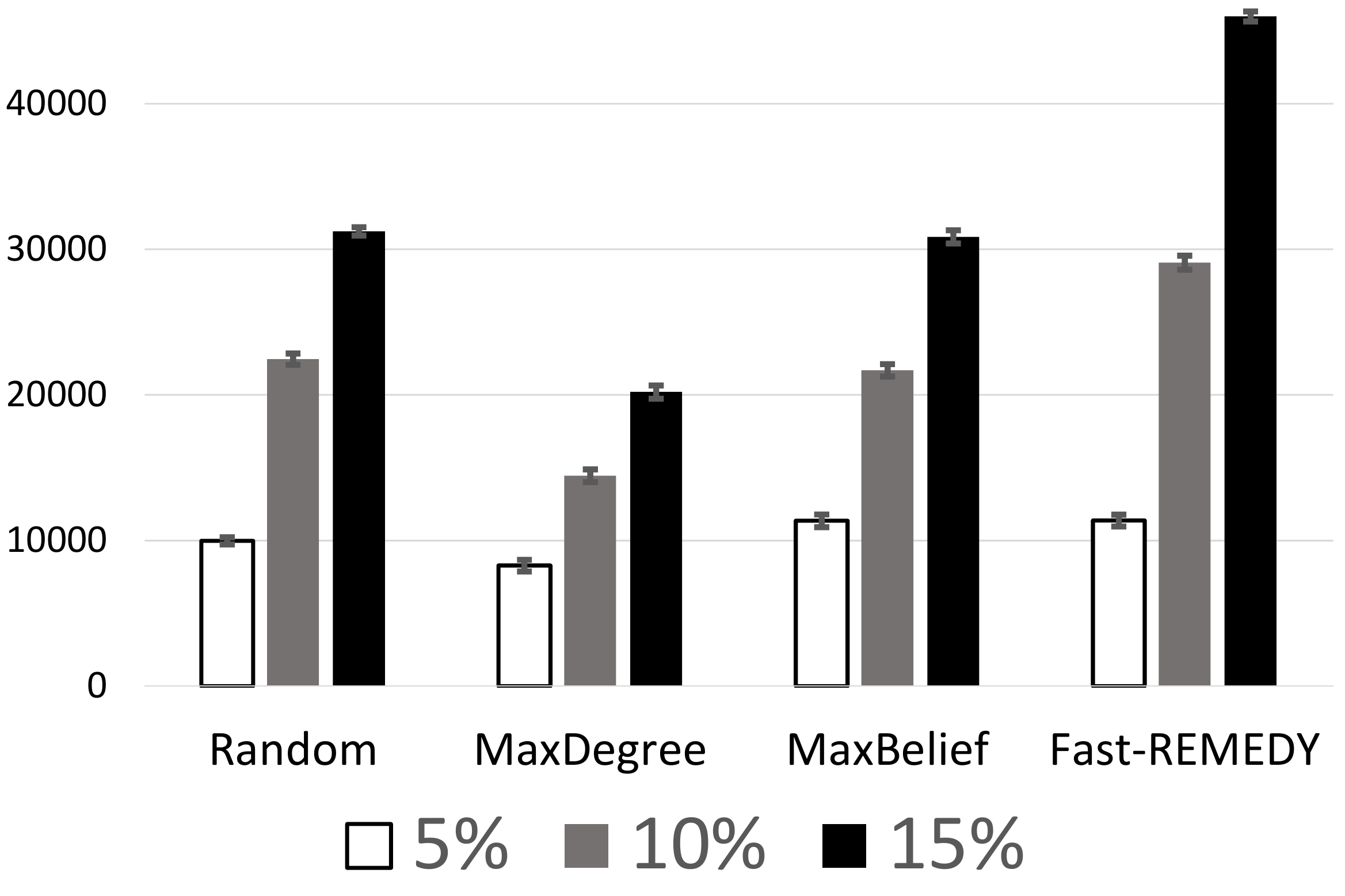}%
}\hfill\null
\centering
\caption{\textbf{Improvement} over \textbf{None} under specific disease parameter of different budget (5\%, 10\%, 15\% of total population respectively).}\label{real}
\end{figure}

%\subsection{Impact of Uncertainty} 

%As section \ref{modelsection} shows, a large portion of our algorithm carefully tackles the uncertainty as the network state is unavailable to us. To evaluate the impact of the uncertainty, we compare the evolution of infected node number over time for different algorithms using the belief state and true state as Figure~\ref{bothdata} shows. 

%\begin{figure}[ht]
%\centering
%\null\hfill
%\subfloat[Belief state]{%
%  \includegraphics[width=0.22\textwidth,keepaspectratio]{./image/Belief_evolve.pdf} %
%}\hfill
%\subfloat[True state]{%
%  \includegraphics[width=0.22\textwidth,keepaspectratio]{./image/True_evolve.pdf} %
%}\hfill\null
%\centering
%\caption{Average number of infected nodes over times in the \textit{Exhibition} network for different algorithm using (a)belief states (b)true states.}\label{bothdata}
%\end{figure}
%When the true state is unavailable, Fast-REMEDY sacrifice a small amount of solution quality for better scalability. Such difference may be small as the uncertainty of our problem canceled out the effect of precise action choice of more complex algorithm. However, it can be observed that in the situation where more information about the network state is available, more complicate algorithm (Full-REMEDY) that considers future action as well has a larger improvement against faster algorithm (Fast-REMEDY). Which algorithm to use really depends on the problem size and scenario of applications.  
\section{Conclusion}
We proposed a novel active screening model (ACTS) to facilitate multi-round active screening problem of SIS recurrent diseases with network structure. We introduced two algorithms, \textsc{Full-REMEDY} and \textsc{Fast-REMEDY} with solution quality and time complexity trade-off and tested on various realistic disease to show their effectiveness.
\bibliographystyle{unsrt} 
\bibliography{template}

\end{document}